\documentclass[final,leqno,onefignum,onetabnum]{siamltex1213}

\usepackage[T1]{fontenc}
\usepackage[utf8]{inputenc}
\usepackage{amsmath,amssymb}
\usepackage{graphicx}
\usepackage[ruled, vlined]{algorithm2e}

\usepackage{comment}

 \usepackage{a4wide}

 \newtheorem{observation}[theorem]{Observation}
 \newtheorem{claim}[theorem]{Claim}

 \newtheorem{problem}[theorem]{Problem}

\newenvironment{proofof}[1]{\par \noindent \textit{Proof of #1.}
}{\hfill$\Box$\medskip}

\newcommand{\idcprob}{\textsc{Min Id Code}}
\newcommand{\scone}{\textsc{Set Cover$_1$}}
\newcommand{\discr}{\textsc{Discriminating Code} }
\newcommand{\id}{\gamma^{\textsc{ID}}}

\title{Identifying codes in hereditary classes of graphs and VC-dimension}

\author{
Nicolas Bousquet \footnotemark[2]\ \footnotemark[6]
\and
Aur\'elie Lagoutte \footnotemark[3]\ \footnotemark[7]
\and
Zhentao Li \footnotemark[4]\ \footnotemark[8]
\and
Aline Parreau \footnotemark[5]\ \footnotemark[9]
\and
St\'ephan Thomass\'e \footnotemark[3]\ \footnotemark[7]}

\begin{document}

\maketitle

\renewcommand{\thefootnote}{\fnsymbol{footnote}}

\footnotetext[2]{LIRMM, Universit\'e Montpellier 2, France}
\footnotetext[3]{LIP, UMR 5668 ENS Lyon - CNRS - UCBL - INRIA, Université de Lyon, France}
\footnotetext[4]{\'Ecole Normale Supérieure, Paris, France}
\footnotetext[5]{LIRIS, UMR 5205, Universit\'e Lyon 1 - CNRS, France}
\footnotetext[6]{Department of Mathematics and Statistics, McGill University and GERAD, Montr\'eal, Canada}
\footnotetext[7]{Partially supported by ANR Project \textsc{Stint} under \textsc{Contract ANR-13-BS02-0007}.}
\footnotetext[8]{Partially supported by a FQRNT B3 postdoctoral fellowship program.}
\footnotetext[9]{Partially supported by a FNRS post-doctoral grant at the University of Li{\`e}ge.}

\renewcommand{\thefootnote}{\arabic{footnote}}

\begin{abstract}

An identifying code of a graph is a subset of its vertices such that every vertex of the graph is uniquely identified by the set of its neighbours within the code.
We show a dichotomy for the size of the smallest identifying code in classes of graphs  closed under induced subgraphs.
Our dichotomy is derived from the VC-dimension of the considered class $\mathcal{C}$, that is the maximum VC-dimension over the hypergraphs formed by the {\bf closed} neighbourhoods of elements of $\mathcal{C}$. We show that hereditary classes with infinite VC-dimension have infinitely many graphs with an identifying code of size logarithmic in the number of vertices while classes with finite VC-dimension have a polynomial lower bound.

We then turn to approximation algorithms. We show that \idcprob\ (the problem of finding a smallest identifying code in a given graph from some class $\mathcal{C}$) is log-APX-hard for any hereditary class of infinite VC-dimension. For  hereditary classes of finite VC-dimension, the only known previous results show that we can approximate \idcprob\ within a constant factor in some particular classes, e.g. line graphs, planar graphs and unit interval graphs. We prove that \idcprob\ can be approximate within a factor $6$ for interval graphs. In contrast, we show that \idcprob\ on $C_4$-free bipartite graphs (a class of finite VC-dimension) cannot be approximated to within a factor of $c \log(|V|)$ for some $c>0$. 
\end{abstract}

\begin{keywords}
Identifying code, VC-dimension, Hereditary class of graphs, Approximation, Interval graph 
\end{keywords}

\begin{AMS}
05C69, 05C85, 05C62
\end{AMS}

\pagestyle{myheadings}
\thispagestyle{plain}
\markboth{Bousquet, Lagoutte, Li, Parreau, Thomass\'e}{IDENTIFYING CODES AND VC-DIMENSION}

\section{Introduction}

Let $G=(V,E)$ be a graph.
An \emph{identifying code} of $G$ is a subset $C$ of vertices
of $G$ such that, for each vertex $v \in V$, the set of vertices in $C$ at
distance at most~1 from $v$, is non-empty and uniquely identifies $v$.
In other words, for each vertex $v\in V(G)$, we have $N[v]\cap
C\neq\emptyset$ ($C$ is a {\em dominating set}) and for each pair $u, v \in V(G)$, we have $N[u]\cap C\neq
N[v]\cap C$ ($C$ is a {\em separating set}), where $N[v]$ denotes the closed neighbourhood of $v$ in
$G$ ($v$ and all its neighbours). We say that a set $X$ of vertices \emph{distinguishes} $u \in V(G)$ from $v \in V(G)$ if $N[u] \cap X \neq N[v] \cap X$. This concept was introduced in 1998 by Karpovsky, Chakrabarty and Levitin~\cite{KCL98} and has applications in various areas such as fault-diagnosis~\cite{KCL98}, routing in networks~\cite{BLT06} or analysis of RNA structures~\cite{Haynes06}. For a complete survey on these results, the reader is referred to  the online bibliography of Lobstein \cite{lob}.

Two vertices $u$ and $v$ are {\em twins} if $N[u]=N[v]$. 
The whole vertex set $V(G)$ is an identifying code if and only if $G$ is twin-free. Since supersets of identifying codes are identifying, an identifying code exists for $G$ if and only if it is twin-free.
A natural problem in the study of identifying codes is to find one of a minimum size. Given a twin-free graph $G$, the smallest size of an identifying code of $G$ is called the \emph{identifying code number} of $G$ and is denoted by $\id(G)$. The problem of determining $\id$ is called the \idcprob\ problem, and its decision version is NP-complete \cite{CHL03}.

Let $X\subseteq V$. We denote by $G[X]$ the graph induced by the subset of vertices $X$. In this paper, we focus on hereditary classes of graphs, that is classes closed under taking induced subgraphs. We consider the two following problems: finding good lower bounds and approximation algorithms for the identifying code number.

\subsection{Previous work}

In the class of all graphs, the best lower bound is $\id(G)\geq \log (|V(G)|+1)$, since all the vertices of the graphs have distinct non-empty neighbourhood within the code. Moncel~\cite{M06} characterized all graphs reaching this lower bound. As for approximation algorithms, the general problem \idcprob\ is known to be $\log$-APX-hard \cite{LT08,BLT06,S07}. In particular, there is no $(1-\varepsilon)\log(|V|)$-approximation algorithm for \idcprob. The problem \idcprob\ remains $\log$-APX-hard even in split graphs, bipartite graphs or co-bipartite graphs (complement of bipartite graphs)~\cite{F13}.

On the positive side, there always exists a $\mathcal{O}(\log{|V(G)|})$ approximation for \idcprob~\cite{S07}. Moreover, even if in the general case \idcprob\ is hard to evaluate, there exist several constant approximation algorithms for restricted classes of graphs, such as planar graphs~\cite{SlaterR84} or line graphs~\cite{FoucaudGNPV13}.

For the remainder of this article, $n$ denotes the number of vertices of $G$.
 Table~\ref{tab:sumup} gives an overview of the currently known results for some restricted  hereditary classes of graphs. The order of magnitude of all lower bounds are best possible (there are infinite families of graphs reaching the lower bounds). \idcprob\ for line graphs and planar graphs have a polynomial time constant factor approximation algorithm with the best known constant written in parenthesis.
From this table, we observe two behaviours: a class either

\begin{enumerate}
\item
  has a logarithmic lower bound on the size of identifying codes, and \idcprob\ is log-APX-hard in this class (for example split, bipartite, co-bipartite graphs), or
\item
 there is a polynomial lower-bound on $\id(G)$ and a constant factor approximation algorithm to compute $\id(G)$.
\end{enumerate}

\begin{table}[h]
\begin{center}
\begin{tabular}{|c|c|c|c|c|}
\hline
Graph class & Lower bound  & Complexity & Approximability & References\\  
\hline \hline
All graphs & $\Theta(\log(n))$ & NP-c & log-APX-hard & \cite{KCL98,LT08} \\  \hline
Chordal & $\Theta(\log(n))$ & NP-c & log-APX-hard & \cite{F13} \\  \hline
Split graphs & $\Theta(\log(n))$ & NP-c & log-APX-hard & \cite{F13} \\  \hline
Bipartite & $\Theta(\log(n))$ & NP-c & log-APX-hard & \cite{F13} \\  \hline
Co-bipartite & $\Theta(\log(n))$ & NP-c & log-APX-hard & \cite{F13} \\ \hline
Claw-free & $\Theta(\log(n))$ & NP-c & log-APX-hard & \cite{F13} \\ \hline
Interval & $\Theta(n^{1/2})$ & NP-c & open & \cite{FoucaudMNPV15-1,FoucaudMNPV15-2}\\  \hline
Unit interval & $\Theta(n)$ & open & PTAS & \cite{Fthese,FoucaudMNPV15-1}\\  \hline
Permutation & $\Theta(n^{1/2})$ & NP-c & open & \cite{FoucaudMNPV15-1, FoucaudMNPV15-2}\\ \hline
Line graphs & $\Theta(n^{1/2})$ & NP-c & APX(4)& \cite{FoucaudGNPV13}\\  \hline
Planar &$\Theta(n)$&NP-c & APX(7) & \cite{ACHL10,SlaterR84}\\ \hline
\end{tabular}
\caption{Known lower bounds on $\id(G)$ and approximability of $\id(G)$.}
\label{tab:sumup}
\end{center}
\end{table}

\subsection{Our results}

The aim of this paper is to shed some light on the validity of such a dichotomy for all classes of graphs using the VC-dimension of the class of graphs.

\paragraph{VC-dimension} Let $\mathcal H=(V, \mathcal E)$ be a hypergraph. A subset $X\subseteq V$ of vertices is {\em shattered} if for every subset $S$ of $X$, there is some hyperedge $e$ such that $e\cap X = S$. The {\em VC-dimension} of $\mathcal H$ is the size of the largest shattered set of $\mathcal H$. 
We define the \emph{VC-dimension of a graph} as the VC-dimension of the closed neighbourhood hypergraph of $G$ (vertices are the vertices of $G$ and hyperedges are the closed neighbourhoods of vertices of $G$), a classical way to define the VC-dimension of a graph (see~\cite{AlonB06,BousquetT12}).

By \emph{a shattered set of a graph $G$}, we mean a shattered set of the hypergraph of the closed neighbourhoods of $G$.
The VC-dimension of a class of graphs $\mathcal C$, denoted by $dim(\mathcal C)$, is the maximum of the VC-dimension of the graphs over $\mathcal C$. If it is unbounded, we say that $\mathcal C$ has \emph{infinite VC-dimension}.

\paragraph{Dichotomy for lower bounds}
First we will prove in Section~\ref{sec:dic} that there is indeed such a dichotomy on the minimum size of identifying codes: it is always either logarithmic or polynomial, where the exponent of the polynomial depends on the VC-dimension of the class of graphs. In particular, our theorem provides new lower bounds for graphs of girth at least $5$, chordal bipartite graphs, unit disk graphs and undirected path graphs. Moreover, these bounds are tight for interval graphs and graphs of girth at least $5$.

\paragraph{Approximation hardness} We then try to extend this dichotomy result for constant factor approximations. First, we show in Section~\ref{sec:approx} that \idcprob\ is $\log$-APX-hard for any hereditary class with a logarithmic lower bound. The proof essentially consists in proving that a hereditary class with infinite VC-dimension contains one of these three classes, for which \idcprob\ has been shown to be $\log$-APX-hard~\cite{F13}: the bipartite graphs, the co-bipartite graphs, or the split graphs. Unfortunately, the dichotomy does not extend to approximation since we show in Section~\ref{sec:approxC4} that $C_4$-free bipartite graphs have a polynomial lower bound on the size of identifying codes but \idcprob\ is not approximable to within a factor $c\log n$ for some $c>0$ (under some complexity assumption) in this class. Thus, a constant factor approximation is not always possible in the second case.
 
\paragraph{Approximation algorithm} Finally, in Section~\ref{sec:interval}, we conclude the paper with some positive result when the lower bound is polynomial by proving that there exists a $6$-approximation algorithm for interval graphs, a problem left open in~\cite{F13}. 

\medskip

The results obtained in this paper are detailed in Table~\ref{tab:results}.

\begin{table}
\begin{center}
\begin{tabular}{|c|c||c||c|c|}
\hline
Graph class & VC dim  & IC-lower bound & 
IC-approx \\
\hline \hline
Girth $\geq 5$ &  $2$ & $\Theta(n^{\frac{1}{2}})$ (opt,new) &open \\ \hline
Interval &2 & $\Theta(n^{\frac{1}{2}})$ (opt) & 6 (Thm.~\ref{thm:interval})\\  \hline
Chordal bipartite & 3  & $\Omega(n^{\frac{1}{3}})$ (new) & open\\
\hline
Unit disk & 3 &$\Omega(n^{\frac{1}{3}})$ (new) & open \\  \hline
$C_4$-free bipartite & 2 & $\Theta(n^{\frac{1}{2}})$ (opt,new) & no $c\log(n)$-approx (Thm. \ref{thm:C4bip}) \\ \hline
Undirected path &3 &$\Omega(n^{\frac{1}{3}})$(new) &open \\ \hline
\end{tabular}
\caption{Overview of the results obtained in this paper.}
\label{tab:results}
\end{center}
\end{table}

\vspace{0.5cm}

\noindent 

\section{Dichotomy for lower bound}\label{sec:dic}

Most of the results using VC-dimension consist in obtaining upper bounds. However, in the last few years, several interesting lower bounds have been obtained using VC-dimension, for instance in game theory (e.g.~\cite{DanielySG14,PapadimitriouSS08}). All these proofs consist in an application of a lemma, due to Sauer \cite{S72} and Shellah \cite{Sh72}, or one of its variants. 
Our result has the same flavour since we use this lemma to prove that the size of an identifying code cannot be too small if the VC-dimension is bounded. 
The \emph{trace} of a set $X$ on $Y$ is $X \cap Y$. By extension, \emph{the trace of a vertex $x$ on $Y$} is the intersection of $N[x]$ with $Y$.

\

\begin{lemma}[Sauer's lemma \cite{S72,Sh72}]\label{lem:sauer}
Let $\mathcal H = (V, \mathcal E )$ be an hypergraph of VC-dimension $d$. For every set $X \subseteq V$, 
the number of (distinct) traces of $\mathcal{E}$ on $X$ is at most $$\sum_{i=0}^{d}{|X| \choose i}\leq |X|^{d}+1.$$ 
\end{lemma}

Let us now prove the main result of this section.

\

\begin{theorem}\label{thm:dic}
For every hereditary class of graphs $\mathcal C$, either
\begin{enumerate}
\item for every $k\in \mathbb N$, there exists a graph $G_k\in \mathcal{C}$ with more than $2^k-1$ vertices and an identifying code of size $2k$, or
\item there exists $\varepsilon > 0$ such that no twin-free graph $G\in \mathcal C$ with $n$ vertices has an identifying code of size smaller than $n^\varepsilon$.
\end{enumerate}
\end{theorem}

\

\begin{proof}
Let $\mathcal C$ be an hereditary class of graphs. The class $\mathcal{C}$ either has finite or infinite VC-dimension.
First, suppose that $\mathcal C$ has infinite VC-dimension. We will show that $\mathcal C$ satisfies the first conclusion. By definition of infinite VC-dimension, there is a graph $H_k \in \mathcal C$ with VC-dimension $k$ for each $k$. So there exists a set of vertices $X$ of size $k$ of $H_k$ which is shattered. Let $Y$ be a set of $2^{k}-1$ vertices whose closed neighbourhoods have all possible traces on $X$ except the empty set, meaning that for every $X'\subseteq X$, add a vertex $y$ in $Y$ such that $N[y]\cap X=X'$. Choose $Y$ so that $|X\cap Y|$ is maximized.
Let $G_k=H_k[X\cup Y]$. The graph $G_k$ has at least $2^k-1$ vertices since $|Y| = 2^{k}-1$.
By choice of $Y$, $X$ dominates $X \cup Y$ and $X$ distinguishes every pair of vertices of $Y$. By maximality of $|X \cap Y|$, $X$ also distinguishes every vertex in $X$ from every vertex in $Y$ (otherwise a vertex of $Y$ would have the same neighbourhood in $X$ as a vertex $x\in X$ and thus can be replaced by $x$, contradicting the maximality of $|X \cap Y|$).
For each $x \in X$, the vertex $y_x \in Y$ whose closed neighbourhood intersects $X$ in exactly $\{x\}$ distinguishes $x$ from all vertices in $X-x$. So $X \cup \{y_x |x \in X\}$ is an identifying code of size at most $2k$, as required.

Now suppose that the VC-dimension of $\mathcal C$ is bounded by $d$. For any identifying code $C$ of a twin-free graph $G \in \mathcal{C}$, the traces of vertices of $G$ on $C$ are different. Hence, by Lemma~\ref{lem:sauer}, $n\leq \sum_{i=0}^d {|C| \choose i} \leq |C|^{d}+1$. Therefore, $|C|\geq (n-1)^{\frac{1}{d}}$, proving that $\mathcal C$ satisfies the second claim.
\qquad \end{proof}

The proof gives in fact the lower bound $\id(G) \in \Omega(n^{\frac{1}{dim(\mathcal C)}})$ for the second item. So if we can bound the VC-dimension of the class, then we immediately obtain lower bounds on the size of identifying codes.
Lemma~\ref{lem:vc} provides such bounds for several classes of graphs.

Let us give some definitions.
The \emph{girth} of a graph is the length of a shortest cycle. A \emph{chordal bipartite graph} is a bipartite graph without induced cycle of length at least $6$. A \emph{unit disk graph} is a graph of intersection of unit disks in the plane. An \emph{interval graph} is a graph of intersection of segments on a line.
An \emph{undirected path graph} is a graph of vertex-intersection of paths in an undirected tree (\emph{i.e.} two vertices are adjacent if their corresponding paths have at least one vertex in common).

\begin{figure}
\begin{center}
\begin{minipage}{0.44\linewidth}
\centering
\includegraphics[height=2.5cm]{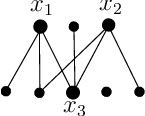}
\caption{The set $x_1,x_2,x_3$ is shattered in this chordal bipartite graph.}
\label{fig:chordalbipartite}
\end{minipage}
\hspace{0.05\linewidth}
\begin{minipage}{0.44\linewidth}
\centering
\includegraphics[height=2.5cm]{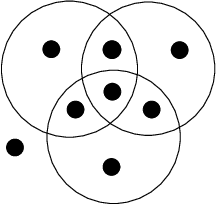}
\caption{A set of three vertices shattered by disks in the plane.}
\label{fig:completevenn}
\end{minipage}
\end{center}
\end{figure}

\

\begin{lemma}\label{lem:vc}
 The following upper bounds hold and are tight:
 \begin{itemize}
  \item The VC-dimension of graphs of girth at least $5$ is at most $2$.
  \item The VC-dimension of chordal bipartite graphs is at most $3$.
 \item The VC-dimension of unit disk graphs is at most $3$.
    \item The VC-dimension of interval graphs is at most $2$.
      \item The VC-dimension of undirected path graphs is at most~$3$.
 \end{itemize}
\end{lemma}

\

\begin{proof}
\begin{itemize}
 \item Let $G$ be a graph of girth at least $5$. Assume by contradiction that a set $\{ x_1,x_2,x_3 \}$ of three vertices is shattered. Since the girth is at least $5$, $x_1x_2x_3$ is not a clique. We may assume without loss of generality that $x_1$ and $x_2$ are not adjacent. Since $\{x_1,x_2,x_3 \}$ is shattered, there is a vertex $y_1$ adjacent to both $x_1$ and $x_2$ and not $x_3$ (one closed neighbourhood must have trace $\{x_1,x_2\}$ on $\{x_1,x_2,x_3\}$) and a vertex $y_2$ adjacent to $\{ x_1,x_2,x_3 \}$ (one closed neighbourhood must have trace $\{ x_1,x_2,x_3 \}$). Note that both $y_1$ and $y_2$ are distinct from $x_1$ and $x_2$ since $x_1$ and $x_2$ are not adjacent. Moreover $y_1$ and $y_2$ are distinct since they do not have the same neighbourhood in $\{x_1,x_2,x_3\}$. So $x_1y_1x_2y_2x_1$ is a cycle of length $4$, a contradiction with the girth assumption. \\
 This bound is tight, for instance with the path on six vertices.
 
 \item Let $G=(A \cup B,E)$ be a chordal bipartite graph. Assume by contradiction that $\{x_1,x_2,x_3,x_4\}$ is a shattered set of four vertices. Since there is a vertex whose closed neighbourhood contains the whole set of vertices, it means that at least three vertices, say $x_1,x_2,x_3$ are on the same side of the bipartite graph. Since a subset of a shattered set is shattered, $\{x_1,x_2,x_3\}$ is shattered. Thus there is a vertex incident to $x_1,x_2$ and not $x_3$, a vertex incident to $x_1,x_3$ and not $x_2$, and a vertex incident to $x_2,x_3$ and not $x_1$. It provides an induced cycle of length $6$, a contradiction.\\
 Moreover the bound is tight, see Figure~\ref{fig:chordalbipartite}.
 
 \item Let $G$ be a unit disk graph. Let us rephrase the adjacency and shattering conditions in this class: let $x_1$ and $x_2$ be any two vertices of a unit disk graph and denote by $c_1$ and $c_2$ their respective centers in a representation of the unit disk graph in the plane. The vertices $x_1$ and $x_2$ are adjacent if and only if $c_1$ and $c_2$ are at distance at most $2$. 
Thus if a set of unit disks is shattered then for every subset of centers, there exists a point at distance at most $2$ from these centers and more than 2 from the others. In other words, there exist points in all possible intersections of balls of radius $2$. \\
 A classical result ensures that the VC-dimension of a hypergraph whose hyperedges can be represented as a set of disks in the plane (and vertices as points of the plane) has VC-dimension at most $3$ (see~\cite{MatousekSW90} for instance). Thus unit disk graphs have VC-dimension at most $3$, and the bound can be reached (see Figure~\ref{fig:completevenn}).

\item Let $G$ be an interval graph. Assume by contradiction that there is a shattered set $\{I_1,I_2,I_3\}$ of $G$.
Assume that $I_1$ starts before $I_2$ and that $I_2$ starts before $I_3$. Since there is an interval $J$ intersecting both $I_1$ and $I_3$ but not $I_2$, $J$ must start after $I_2$ and thus $I_1$ contains $I_2$. Then there is no interval intersecting $I_2$ but not $I_1$, a contradiction. Thus interval graphs have VC-dimension at most $2$, and the bound is again reached with the path on six vertices.

 \begin{figure}
 \begin{center}
  \includegraphics[scale=1.2]{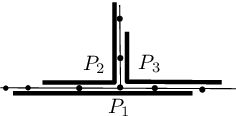}
  \caption{Paths $P_1,P_2,P_3$ are shattered by the eight points which are paths of length $0$.}
  \label{fig:undirectedpath}
 \end{center}
 \end{figure}
 
 \item Let $\mathcal{P}=\{P_1,P_2,P_3,P_4\}$ be a shattered set of four paths of a tree $T$. Assume first that $P_2, P_3, P_4$ all intersect $P_1$ and consider the restriction of $T$ to $P_1$, which is in fact an interval graph. To ensure all possible intersections with $P_1$, the set $\{P_2, P_3, P_4\}$ is a shattered set of size three in an interval graph, a contradiction. 
 
Thus at least one path, say $P_2$, does not intersect $P_1$ and lies in a connected component $C$ of the forest $F=T\setminus P_1$. If $P_3$ does not intersect $C$, then there is no path intersecting both $P_2$ and $P_3$ but not $P_1$. Thus $P_3$ intersects $C$. If moreover $P_3$ intersects $P_1$, then no path can intersect both $P_1$ and $P_2$ but not $P_3$. Thus $P_3$ is also included in $C$. Let $P$ be a path intersecting $P_1$, $P_2$ and $P_3$. Assume first that $P$ intersects the three paths in the order $P_1$, $P_2$ and $P_3$ (the case $P_1$, $P_3$, $P_2$) is the same. Then no path can intersect $P_1$ and $P_3$ without intersecting $P_2$. Assume now that $P$ intersects the three paths in the order $P_2$, $P_1$, $P_3$. Similarly, no path can intersect $P_2$ and $P_3$ without intersecting $P_1$. Hence the path $P$ cannot exist, a contradiction.
Finally the bound of $3$ can be reached, as shown in Figure~\ref{fig:undirectedpath}.
\end{itemize}
\qquad \end{proof}

Lemma~\ref{lem:vc} and Theorem~\ref{thm:dic} imply new lower bounds for many classes: 
$\Omega(n^{\frac{1}{2}})$ for graphs with girth at least $5$,
$\Omega(n^{\frac{1}{3}})$ for chordal bipartite graphs,
$\Omega(n^{\frac{1}{3}})$ for unit-disk graphs,
$\Omega(n^{\frac{1}{2}})$ for interval graphs,
$\Omega(n^{\frac{1}{3}})$ for permutation graphs, 
and
$\Omega(n^{\frac{1}{3}})$ for undirected path graphs.

The exponent given by Theorem~\ref{thm:dic} is sharp for several classes of graphs. Indeed, 
Foucaud et \emph{al.} \cite{FoucaudMNPV15-1} proved that there are infintely many interval graphs with identifying codes of size $\Theta(n^{1/2})$. The bound is also tight for $C_4$-free bipartite graphs (which have girth at least $5$): the following construction is a $C_4$-free bipartite graphs with an identifying code of size $\Theta(n^{\frac{1}{2}})$. Let $G=(X \cup Y,E)$ be a bipartite graph where $Y$ has size $n$, $X$ has size $\frac{n(n-1)}{2}$, and edges satisfy the following rule: for every pair $u,v$ of vertices of $Y$, there is exactly one vertex of $X$ adjacent to both $u$ and $v$. The graph $G$ does not contain any triangle (since it is bipartite) nor $C_4$ (since neighbourhoods intersect on at most one vertex). One can easily check that the set $Y$ is an identifying code of the graph. Indeed vertices of $X$ are adjacent to precisely two neighbours on $Y$ and vertices of $Y$ have precisely one neighbour on $Y$ in their closed neighbourhood. Finally, it is also sharp for the class of all graphs of VC-dimension at most $d$. Indeed, consider the bipartite graph made with a stable set $A$ of size $d$ and a stable set 
$B$ of size $\sum_{i=2}^d {d \choose i}$ representing all the subsets of $A$ of size at least $2$. Each vertex of $B$ is adjacent to the vertices of $A$ corresponding to its subset. This graph has VC-dimension $d$ and the set $A$ is an identifying code of size of order $n^{1/d}$.

Nevertheless, the bounds given by Theorem~\ref{thm:dic} are not necessarily tight. For instance, permutations graphs can have VC-dimension $3$ but Foucaud et \emph{al.} \cite{FoucaudMNPV15-1} recently proved that the exact lower bound is $\Omega(n^{\frac{1}{2}})$.


\section{Inapproximability in infinite VC-dimension}\label{sec:approx}

Given a minimization problem $P$ and a function $f:\mathbb N\to \mathbb N$, a  \emph{factor} $f$ \emph{approximation algorithm} (also called an \emph{$f$-approximation}) is an algorithm that outputs a solution of value at most $f(n) \cdot OPT(I)$ for every instance $I$ of $P$ of size $n$, where $OPT(I)$ is the value of an optimal solution of $I$. The class \emph{$\log$-APX} is a class of problems consisting of all problems that admit a logarithmic factor polynomial time approximation algorithm. We use the AP-reductions introduced in \cite{CKST} which have now become standard. Its definition restricted to minimization problems is defined as follows:

\

\begin{definition}[\cite{ACGKMP}]
Let $P$ and $Q$ be two minimization problems. An \emph{AP-reduction} from $P$ to $Q$ is a triple $(f,g,\alpha)$ where

\begin{enumerate}
\item
$\alpha$ is a constant,
\item
$f$ maps pairs consisting of an instance of $P$ and a constant $r>1$ to instances of $Q$, and
\item
$g$ maps triples consisting of a constant $r>1$, an instance $I_P$ of $P$ and a solution to $f(I_P,r)$ to a solution of $I_P$
\end{enumerate}

 in such a way that

\begin{enumerate}
\item
$f(I_P, r)$ has a solution if $I_P$ does,
\item
$f(\cdot, r)$ and $g(\cdot, \cdot, r)$ are computable in polynomial time for all fixed $r$, and
\item
if $SOL_Q$ is a solution of $f(I_P,r)$ of size at most $r \cdot OPT(f(I_P,r))$, then the solution $g(f(I_P,r), r, SOL_Q)$ has size at most $(1+\alpha (r-1)) \cdot OPT(I_P)$.
\end{enumerate}
\end{definition}

\

A problem $Q$ is \emph{$\log$-APX-hard} if any problem $P$ in $\log$-APX can be reduced to $Q$ by an AP-reduction. 

\

\begin{theorem}[\cite{CKST}]
Any optimization problem $P$ that is $\log$-APX-hard with respect to AP-reduction is NP-hard to approximate within a factor $c \cdot \log(n)$ where $n$ is the size of the input, for some constant $c > 0$.
\end{theorem}

\

We show that \idcprob\ is $\log$-APX-hard for classes with infinite VC-dimension.
To prove this result, we will prove that a class with infinite VC-dimension contains either all the bipartite graphs, or all the co-bipartite graphs or all the split graphs. Since the problem \idcprob\ is $\log$-APX-hard in these three classes (see \cite{F13}), it implies that it is $\log$-APX-hard for all classes with infinite VC-dimension.

\

\begin{theorem}\label{thm:allbip}
Let $\mathcal C$ be an hereditary class. If $\mathcal C$ has infinite VC-dimension, then $\mathcal C$ must contain either all the bipartite graphs, or all the co-bipartite graphs or all the split graphs.
\end{theorem}

\

Note that this result implies the first part of Theorem \ref{thm:dic}.
  We say that a bipartite graph $H=(A \cup B, E)$ is a \emph{bipartisation} of $G$ if removing all edges in $A$ and in $B$ in $G$ yields $H$ for some partition $A,B$ of $V(G)$.

\

 \begin{lemma}\label{lem:fullcrossing}
   For any  hereditary class $\mathcal{C}$ of graphs with infinite VC-dimension and any bipartite graph $H$, $\mathcal{C}$ contains a graph $G$ whose bipartisation is $H$.
 \end{lemma}
 
\

\begin{proof}
  Let $H=(A\cup B,E)$ be a bipartite graph with $|B| \le |A|=k$.
  Since $\mathcal C$ has infinite VC-dimension, it contains a graph $G$ with a shattered set $S$ of size (at least) $\ell = k+\lceil \log(2k) \rceil$. Let $A'$ be the first $k$ vertices in $S$, and let us number $Y_1, \ldots, Y_{2k}$ some  $2k$ distinct subsets of $S\setminus A'$ (they exists since $|S \setminus A'| = \lceil \log(2k) \rceil$).

  By definition of a shattered set, for each $i\in \{1, \ldots 2k\}$ and for each $X \subseteq A'$ there is a vertex $x_i$ of $G$ such that $N[x_i]\cap S=X\cup Y_i$. Thus there are $2k$ vertices of $G$ whose closed neighbourhoods intersect $A'$ in exactly $X$. Hence there are at least $k$ such vertices in $V(G)\setminus A'$. Label the vertices of $A'$ by vertices in $A$, \emph{i.e.} choose an arbitrary bijection between $A$ and $A'$. Now for each $b\in B$, choosing $X=N(b)$ gives $k$ vertices in $V(G)\setminus A'$ whose closed neighbourhoods intersect $A'$ in exactly $N(b)$. So we can choose one "representative" for each $b$  so that all the selected vertices are distinct (since $|B| \leq k$). Note that we need $k$ vertices in $V \setminus A'$ since up to $|B|$ vertices of $B$ may have the same neighbourhood in $A$.
  
  Since $\mathcal{C}$ is closed under taking induced subgraphs, the subgraph of $G$ induced by $A'$ and the set $B'$ of all chosen vertices is in $\mathcal{C}$. The bipartisation of this graph is $H$, as required. 
\qquad \end{proof}
 
\vspace{0.3cm}

Next we show that we can further restrict $H'$ and now require both sides of $H'$ to be stable sets or cliques.
For a bipartite graph $H=(A\cup B,E)$, write 
$H^{1,0}$ for the graph obtained from $H$ by adding a clique on $A$, 
$H^{0,1}$ the graph obtained from $H$ by adding a clique on $B$ and
$H^{1,1}$ the graph obtained from $H$ by adding a clique on both $A$ and $B$. We also write sometimes $H^{0,0}$ for $H$. We show that, for each bipartite graph $H$, $\mathcal{C}$ contains one of these four graphs.
To do so, we need the classical theorem of Erd\H{o}s and Hajnal~\cite{EH} as well as its bipartite version by Erd\H{o}s, Hajnal and Pach~\cite{EHP}.

\

\begin{theorem}[Erd\H{o}s, Hajnal~\cite{EH}]\label{EH}
  For every graph $H$, there exists a constant $c(H)$ such that all graphs on $n$ vertices contain either $H$ as an induced subgraph, a stable set of size at least $2^{c(H)\sqrt{2\log n}}$ or a clique of size at least $2^{c(H)\sqrt{2 \log n}}$.
\end{theorem}

\

\begin{theorem}[Erd\H{o}s, Hajnal, Pach~\cite{EHP}]\label{EHP}
  Let $H$ be a bipartite graph with vertex classes $U_1$ and $U_2$,
  ($k=|U_1| \le |U_2 | = \ell$) and let $n > \ell^{k+1}$. 
  Then in any bipartite graph $G$ with vertex classes $V_1$ and $V_2$ ($|V_1| =
  |V_2| = n$) which contains no two subsets $U_1 \subseteq V_1$, $U_2 \subseteq V_2$ that induce an isomorphic copy of $H$, there exist $V_1'\subseteq V_1$ and $V_2' \subseteq V_2$ of size $\left \lfloor \left(\frac{n}{\ell}\right)^{\frac{1}{k}} \right\rfloor$ such that either all edges between $V'_1$ and $V'_2$ belong to $G$ or none of them does.
\end{theorem}

\

We continue with the following technical lemmata:

\

\begin{lemma}\label{lem:largebipartite}
For $n$ large enough, there exists a bipartite graph $G_0=(A\cup B,E_0)$ with $2n$ vertices ($|A|=|B|=n$) such that there is no complete nor empty bipartite graphs $G_0[A'\cup B']$ with $A'\subseteq A$ and $B'\subseteq B$ and $|A'|=|B'|=\lfloor 2 \log n \rfloor$.
\end{lemma}

\

\begin{proof}
Let us show its existence with a probabilistic argument. Let $A$ and $B$ be two stable sets each of size $n$ and for every $a\in A, b\in B$, put the edge $ab$ with probability $\frac{1}{2}$. 
Given two subsets $A'\subseteq A$, $B'\subseteq B$ with $|A'|=|B'|=\lfloor 2\log n \rfloor$, the probability that $A'\cup B'$ induces a complete bipartite graph is $\left(\frac{1}{2}\right)^{\lfloor2\log n\rfloor^2}$.
 The same probability holds for $A'\cup B'$ inducing an empty bipartite graph. Thus the probability that there exists a complete or empty bipartite graph with each part of size $2\log n$ is at most 

$${n \choose \lfloor 2\log n \rfloor}^2\frac{2}{2^{\lfloor 2 \log n\rfloor^2}} \leq  \left(\frac{n\cdot e}{\lfloor2\log n\rfloor}\right)^{2 \cdot \lfloor2\log n\rfloor}\cdot \frac{2}{2^{\lfloor2\log n\rfloor^2}}=2^{-4\log n \cdot \log \log n +\mathcal{O}(\log n)}$$ 

using the inequality

$${n \choose l} \leq \frac{n^l}{l!} \leq \left(\frac{n\cdot e}{l}\right)^l$$

This probability is strictly less than 1 for $n$ large enough, so there exists a graph $G_0=(A\cup B,E_0)$ for which the event does not occur.
\qquad \end{proof}

\

\begin{lemma}\label{lem:oneoffour}
Let $\mathcal C$ be an hereditary class with infinite VC-dimension. For any bipartite graph $H=(H_\ell\cup H_r, E)$, one of the four graphs $H^{0,0}$, $H^{1,0}$, $H^{0,1}$ or $H^{1,1}$ is in $\mathcal C$.
\end{lemma}

\

\begin{proof}
Suppose by contradiction that Lemma~\ref{lem:oneoffour} is false for $H$ with $|H_\ell| \le |H_r|=k$. Let $c(H)$ be the constant from Theorem \ref{EH} and pick $n$ large enough so that $2^{c(H)\sqrt{2\log n}}>k^{k+1}$, and $(2^{\frac{c(H)}{k}\sqrt{2\log n}})/k^{\frac{1}{k}} > 2\log n$ and $n$ satisfies the condition of Lemma~\ref{lem:largebipartite}. Let $G_0$ be a bipartite graph as in Lemma \ref{lem:largebipartite}, \emph{i.e.} $G_0$ has $n$ vertices on both sides and does not contain a complete or an empty bipartite graph with $2\log n$ vertices on each side.
By Lemma \ref{lem:fullcrossing}, $\mathcal C$ contains a graph $G$ whose bipartisation is $G_0$. Let $A,B$ certify this bipartisation.

Since $G$ contains no copy of $H^{0,0}$, neither does $G[A]$.
So by Theorem \ref{EH}, $G[A]$ contains a clique or stable set $A'$ of size at least $n' = 2^{c(H)\sqrt{2\log n}}$. Similarly, $G[B]$ also contains a clique or stable set $B'$ of this size. Assume that $A'$ and $B'$ induce stable sets (respectively, $A'$ induce a stable set and $B'$ a clique\footnote{The case with $A'$ a clique and $B'$ a stable set is symmetric.} and $A'$ and $B'$ induce cliques). By assumption, since the class $\mathcal{C}$ is closed under induced subgraphs, $G[A'\cup B']$ contains no copy of $H^{0,0}$ (respectively, $H^{1,0}$ and $H^{1,1}$). Hence the bipartisation of $G[A' \cup B']$ contains no copy of $H$. So by Theorem \ref{EHP} and since $n'>k^{k+1}$, the bipartisation of $G[A' \cup B']$ contains a complete bipartite graph or an empty bipartite graph where each bipartition has size
\[
 \left(\frac{n'}{k}\right)^{\frac{1}{k}} = \frac{2^{\frac{c(H)}{k}\sqrt{2\log n}}}{k^{\frac{1}{k}}} > 2\log n
\]
which is a contradiction to $G_0$ having no such subgraph.
\qquad \end{proof}

\

\begin{proofof}{Theorem \ref{thm:allbip}}
Let $H_n$ be the disjoint union of every bipartite graphs of size at most $n$.
For every $n$, Lemma~\ref{lem:oneoffour} ensures that $H_n^{a_n,b_n}$ is in $\mathcal{C}$ (for some $a_n, b_n \in \{0,1\}$) and hence there exist $a,b\in \{0,1\}$ for which $H_n^{a,b}$ is in $\mathcal{C}$ for infinitely many values of $n$. 

If $a=b=0$, all bipartite graphs are in $\mathcal C$; if $a\neq b$, all split graphs are in $\mathcal C$ and if $a=b=1$, all co-bipartite graphs are in $\mathcal C$: indeed let $H^{a,b}$ be a bipartite graph on $n$ vertices (resp. split graph, co-bipartite graph, depending on the value of $a$ and $b$). Then there exists $n'\geq n$ such that $H_{n'}^{a,b}$ is in $\mathcal{C}$. But $H^{a,b}$ is an induced subgraph of $H_{n'}$ so $H^{a,b}$ is an induced subgraph of $H_{n'}^{a,b}$. The theorem follows.
\end{proofof}

\

Foucaud \cite{F13} proved that \idcprob\ is $\log$-APX-hard for bipartite graphs, split graphs and co-bipartite graphs. So the following is a direct corollary of Theorem~\ref{thm:allbip}.

\

\begin{corollary}
\idcprob\ is $\log$-APX-hard when the input graph is restricted to an hereditary class of graphs with infinite VC-dimension.
\end{corollary}

\section{Inapproximability for $C_4$-free bipartite graphs}
\label{sec:approxC4}

In this section, we examine the complexity of approximating \idcprob\ in classes of finite VC-dimension. Previous results suggest that all these classes may have a constant factor approximation algorithm : this is the case for line graphs~\cite{FoucaudGNPV13}, planar graphs~\cite{SlaterR84} or unit interval graphs (since any solution has size at least $\frac{n}{2}$) for instance.

However, we show that this intuition is false: the class $\mathcal{C}$ of $C_4$-free bipartite graphs (whose VC-dimension is bounded by 2) does not admit such an approximation algorithm. In fact, \idcprob\ in $\mathcal{C}$ is hard to approximate to within a $c \log n$ factor (for some $c>0$) in polynomial time, unless $NP\subseteq ZTIME(n^{O(\log \log n)})$.

\

\begin{observation}\label{obs:vcc4free}
The class of $C_4$-free bipartite graphs has VC-dimension at most $2$.
\end{observation}

\

\begin{proof}
Let $G$ be a $C_4$-free bipartite graph. Then it has no triangle and no $C_4$, so we can apply the result of Lemma~\ref{lem:vc} for graphs of girth at least $5$.
\qquad \end{proof}

\

We provide a polynomial time gap preserving reduction (in fact, an AP-reduction) from the following minimization problem:

\

\begin{problem}
\textsc{Set cover with intersection $1$} (\scone)\\
\textbf{Instance:} A set $X$ and a family $S$ of subsets of $X$ where any two sets in $S$ intersect in at most one element.\\
\textbf{Solution:} A subset $S'$ of sets in $S$ whose union contain $X$.\\
\textbf{Measure:} The size of $S'$.
\end{problem}

\

Anil Kumar, Arya and Hariharan~\cite{KAR} have shown that this problem cannot be approximated to within a $c \log n$ factor (for some $c>0$) in polynomial time, unless $NP\subseteq ZTIME(n^{O(\log \log n)})$.

\

\begin{theorem}\label{thm:C4bip}
  \idcprob\ with input restricted to $C_4$-free bipartite graphs cannot be  approximated to within a $c \log n$ factor (for some $c>0$) in polynomial time, unless $NP\subseteq ZTIME(n^{O(\log \log n)})$ where $n$ is the size of the input.
\end{theorem}

\

To give a flavour of our reduction from \scone  \ to \idcprob, we first give an easier reduction to the \textsc{Discriminating code} problem \cite{CCCH06}. The \textsc{Discriminating code} is often a way to design reductions which gives an overview of most complicated ones for \idcprob\ : indeed a discriminating code consists in identifying vertices of a set $X$ using vertices of a set $Y$.

\

\begin{problem}
\textsc{Discriminating code}\\
\textbf{Instance:} A bipartite graph $G=(X \cup Y,E)$.\\
\textbf{Solution:} A subset $Y'$ of $Y$ which dominates $X$ and such that for every pair of vertices $x_1,x_2$ of $X$, $N[x_1]\cap Y' \neq N[x_2] \cap Y'$. Such a set is called a discriminating code.\\
\textbf{Measure:} The size of $Y'$.
\end{problem}

\

\begin{lemma}\label{lem:discr}
  \textsc{Discriminating code} with input restricted to $C_4$-free bipartite graphs cannot be approximated to within a $c \log n$ factor (for some $c>0$) in polynomial time, unless $NP\subseteq ZTIME(n^{O(\log \log n)})$.
\end{lemma}

\

\begin{proof}
Let $ I_{SC}=(X,S)$ be an instance of \scone. The proof is decomposed into five steps: construct an instance $I_{DC}$ of discriminating code that has polynomial size in $|I_{SC}|$; check that this instance is indeed a $C_4$-free bipartite graph; for every solution of $I_{SC}$, construct a solution of $I_{DC}$; and vice-versa; finally check that if the solution of $I_{DC}$ is not too big with respect to the optimal one, then so is the solution of $I_{SC}$. 

\

\paragraph{Construct the instance of \discr}
Let $G=(X\cup S, E)$ be the membership bipartite graph of the instance $I_{SC}$, that is to say that for every $x\in X, s\in S$, there is an edge $xs\in E$ if and only if $x\in s$. In the following, $n$ denotes the size of $X$ and we assume that $n\geq 2$ and that no $s\in S$  is connected to all of $X$ (meaning that the optimal solution to $I_{SC}$ has size at least 2). Note that in the other case, we can compute the optimal solution in polynomial time. Moreover, we assume that for every $x\in X$, there exists $s\in S$ such that $x$ belong to $s$, otherwise there is no solution. The following construction is illustrated on Figure~\ref{fig:C4bipartite}. Let $G_1, \ldots, G_{\ell}$ be $\ell=2n^2-1$ disjoint copies of $G$. Denote by $X_i\cup S_i$ the $i$-th copy of $X\cup S$.
Let $X'_1$ and $X'_2$ be  copies of $X$. For each $x''\in X'_2$, add an edge between $x''$ and its copy $x'$ in $X'_1$, and add edges between $x''$ and its copies in all $G_i$ for $i \leq \ell$. In other words, $G[X_i \cup X_2']$ induces a matching for every $i$.
Let $G_{DC}$ be this bipartite graph with parts $X_{DC}=X_1\cup ...\cup X_{\ell}\cup X'_1$ and $Y_{DC}=S_1\cup ... \cup S_{\ell}\cup X'_2$. Clearly, the size of $I_{DC}$ is polynomial in $n$ an thus in the size of $I_{SC}$.

\

\paragraph{Check that the instance is $C_4$-free}
First note that the initial graph $G$ is $C_4$-free. Indeed every $C_4$ must have two vertices in $S$ and two vertices in $C$, a contradiction since the neighbourhoods of two vertices of $S$ intersect on at most one vertex.  Further, the graph $G_{DC}$ is $C_4$-free. Indeed, no $C_4$ can contain two vertices of $S_1\cup \dots \cup S_l$ since any vertex $s\in S_i$ only has neighbours in $X_i$, and $G_{DC}[X_i \cup S_i]$ is a copy of the $C_4$-free graph $G$. Moreover, two vertices $x''\in X_2'$ and $s\in S_i$ have at most one common neighbour $x_i\in X_i$, the copy of $x''$. Finally, each pair of vertices of $X_2'$ have disjoint neighbourhoods. Thus no vertex can be part of a $C_4$.

\

\paragraph{Transforming a solution of $I_{SC}$ into a solution of $I_{DC}$}
Let $D$ be a set cover of $S$ of size $SOL_{SC}$. Construct $C$ as the union of $\ell$ copies of $D$ (one for each $G_i$), and finally add $X'_2$. Then $C$ is dominating $X_{DC}$ since $X'_2$ is, moreover $C$ is separating all the pair of vertices of $X_{DC}$. Indeed, two vertices $x_i, y_j$ inherited from two different elements $x,y\in X$ are separated by $x''\in X'_2$, the copy of $x$ in $X'_2$. Two vertices $x_i \in X_i, x_j\in X_j$ with $i\neq j$, or two vertices $x_i\in X_i, x'\in  X'_1$ inherited from the same element $x\in X$ are separated by a neighbour $s_i\in S_i$ of $x_i$, where $s_i$ is the $i$-th copy of an element $s\in D$ containing $x$. 
Consequently, $C$ is indeed a discriminating code. 
For later use, observe that if $D$ is the optimal solution of \scone \ of size  $OPT_{SC}$, we can derive $OPT_{DC}\leq n+ \ell \cdot OPT_{SC}$. 

\

\paragraph{Transforming a solution of $I_{DC}$ into a solution of $I_{SC}$}
Let $C$ be a solution of $I_{DC}$ of size $SOL_{DC}$. We construct a set cover candidate $D_1=S_1\cap C$. Every vertex $x_1\in X_1$ is dominated and separated from its copy $x'$ in $X'_1$, and $x''\in X'_2$ cannot achieve this goal, thus there exists $s\in S_1\cap C$ which is linked to $x_1$. Thus $D_1$ is a set cover. The same can be done for each $i\leq \ell$, and we choose $D$ as the minimum size such constructed set cover. For later use, observe that $X_2'\subseteq C$ since $X_1'$ is dominated. Consequently if $SOL_{SC}$ is the size of $D$, we have $|X_2'| + \ell \cdot SOL_{SC} \leq SOL_{DC}$ or equivalently $SOL_{SC} \leq \frac{SOL_{DC} -n}{\ell}$. 

\

\paragraph{Concluding on the size of the solutions} Now suppose that we can obtain a solution $SOL_{DC}$ of \discr satisfying $SOL_{DC} \leq r \cdot OPT_{DC}$ for some value $r$. The above discussion gives

$$SOL_{SC}\leq \frac{SOL_{DC}-n}{\ell} \leq \frac{r \cdot OPT_{DC}-n}{\ell} \leq \frac{r(n+\ell \cdot OPT_{SC})-n}{\ell}\leq 2 r \cdot OPT_{SC}$$

In particular if $r=c' \log n$ for some well-chosen constant $c'$, we obtain a contradiction with \scone \ approximation hardness.
\qquad \end{proof}

\

 \begin{figure}
\begin{center}
\vspace{10pt}
\includegraphics{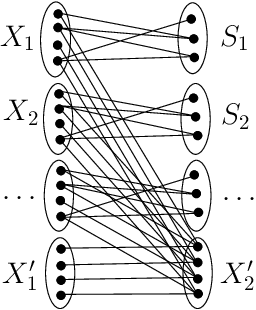}
\end{center}
\caption{Construction of the proof of Lemma~\ref{lem:discr}.}
\label{fig:C4bipartite}
\end{figure}

Let us now adapt this reduction into a reduction to identifying codes.

\

\begin{proofof}{Theorem~\ref{thm:C4bip}}
Let $I_{SC}$ be an instance of \scone.

\

\paragraph{Construct the instance of \idcprob} First construct the same graph $G_{DC}$ as in the proof of Lemma~\ref{lem:discr}. Now for identifying codes, we need to identify vertices in both sides and not only on the side of $X_{DC}$.
For that, we add to $G_{DC}$ a set $Z=\{z_1, \ldots, z_{2n^2}\}$ in part $X_{DC}$. We have to be careful when we connect the vertices of $Z$ to the vertices of the graph since we do not want to create a $C_4$. We aim at choosing edges between $Z$ and $S_1\cup ... \cup S_{\ell}$ such that each vertex $s\in S_1\cup ... \cup S_{\ell}$ is adjacent to exactly two vertices of $Z$, and no two vertices $s,s'\in S_i$ share a neighbour in $Z$. The following claim (whose proof is postponed at the end of the section) reaches the goal:

\

\begin{claim} \label{C4-free}
There exists a numbering of the vertices in $S_1\cup ... \cup S_{\ell}$ such that:
\begin{itemize}
\item Each vertex $s\in S_1\cup ... \cup S_{\ell}$ is numbered $s_{i,j}$ with $i<j\in \{1, \ldots, 2n^2\}$, where the pair $\{i,j\}$ is distinct for every vertex.
\item Two vertices $s_{j,k}$ and $s_{j',k'}$ cannot both belong to the same set $S_i$ if one of $j,k$ is equal to one of $j',k'$.
\end{itemize}
\end{claim}

\

Using the numbering of the claim, we just have to add the edges $z_ks_{k,l}$ and $z_ls_{k,l}$ for every $k,l\in \{1, \ldots, 2n^{2}\}$. Note that every vertex of $S_i$ is connected to precisely two vertices of $Z$, and that every vertex $z\in Z$ has at most one neighbour in each $S_i$. Let $G_{IC}$ be this new graph. It has polynomial size in $n$.

\

\paragraph{Check that the instance is $C_4$-free}  Since we only add edges from $Z\subset X_{DC}$ to $S_1 \cup \dots \cup S_\ell \subset Y_{DC}$, the graph is indeed bipartite. Since $G_{DC}$ was $C_4$-free, any hypothetical $C_4$ must intersect $Z$, say in $z_k \in Z$. $z_k$ share with any $z_l\in Z$ at most one common neighbour $s_{k,l}$, so the $C_4$ must intersect $X_i$ for some $i$, or $X'_1$. On the one hand, vertices in $X'_1$ have degree one. On the other hand, $x_i$ only has neighbours in $S_i$, and $z_k$ has one only neighbour in each $S_i$, so they cannot be in the same $C_4$.

\

\paragraph{Transforming a solution of $I_{SC}$ into a solution of $I_{IC}$}
A set $C$ containing $X_1 \cup X'_2 \cup Z$  is a good candidate to be an identifying code because it has the following properties:
\begin{itemize}
\item For every $z_k\in Z$, $z_k$ is identified by $z_k$ being the only vertex of $Z\cap N[z_k] \cap C$. 
\item For every copy $x''\in X'_2$ of an element $x\in X$, $x''$ is dominated by $\{x'', x_1\}$ in $C$ a where $x_1\in X_1$ is the copy of $x$. Thus it is separated from all the other vertices except maybe $x_1$.
\item For every $s_{k,l}\in S_1\cup \ldots \cup S_{\ell}$, $s_{k,l}$ is identified by $\{z_k, z_l\}$.
\item For every $x'\in X_1'$, $x'$ is dominated by $x''\in X_2'$.
\item For every $x_i\in X_i$, $x_i$ is dominated by $x''\in X_2'$
\end{itemize}

Thus $C$ is a dominating set, and the only sets of vertices that may be not separated are of the form $\{x'\}\cup \{x_2, \ldots, x_{\ell}\}$  and $\{x'',x_1\}$ for any element $x\in X$.

Let $D$ be a set cover of the initial instance and $D_1,...,D_{\ell}$ be the respective copies in the graphs $G_i$. Then  $C=D_1\cup...\cup D_{\ell}\cup X_1 \cup X'_2\cup Z$ is an identifying code of $G_{IC}$. Indeed, every vertex $x_i$ is separated from $x'$, $x''$ and $x_j$ ($i\neq j$) by the element of $S_i$ that covers it in the set cover $D_i$.
Hence any solution for \scone \ of size $SOL_{SC}$ gives a solution for \idcprob\ of size $SOL_{IC}=\ell \cdot SOL_{SC}+2n+2n^2 $. In particular $OPT_{IC}\leq \ell \cdot OPT_{SC}+2n+ 2n^2$.

\

\paragraph{Transforming a solution of $I_{IC}$ into a solution of $I_{SC}$}
Let $C$ be an identifying code of $G_{IC}$. 
We define $D_i=C\cap S_i$ as a set cover candidate. Unfortunately, $D_i$ may not be a set cover, in which case we iteratively modify $C$ until all $D_i$ meet the condition, starting with $D_1$.
If $D_i$ is not a set cover of $X_i$, then there is a vertex $x_i$ not covered. This vertex must be separated with its copy $x'$ in $X'_1$, hence $x_i$ (case 1) or $x'$ (case 2) must belong to $C$ (if both occur, case 1 has priority on case 2). Then choose any neighbour $s\in S_i$ of $x_i$, add this vertex to $C$ and remove $x_i$ (in case 1) or $x'$ (in case 2). We thus get a new set $C'$ and claim that $C'\cup Z\cup X_1\cup X_2'$ is an identifying code. Thanks to the above discussion, we just have to show that the sets $\{x'\}\cup \{x_2, \ldots, x_{\ell}\}$ and $\{x'',x_1\}$ are separated.

Observe first that $x_i$ is now separated from $x'$ and from $x_j$ by $s$ for $j\neq i$. Moreover, since $C$ was separating $x_{j_1}$ from $x_{j_2}$  for $j_1, j_2\neq i$, then $C'$ still does (because nothing changed in their neighbourhood). We also have $C'$ that separates $x'$ from $x_j$ for $j\neq i$: the vertex separating those two vertices was not $x_i$, so in case 1 it still belongs to $C'$. In case 2, we have removed $x'$ but then $x_i$ was not in $C$, and $C$ was separating $x_i$ and $x_j$ so there exists a vertex in $(N[x_j]\setminus \{x'\})\cap C$, and this vertex separates $x_j$ from $x'$ in $C'$. Finally, $C'$ separates $x_1$ from $x''$ since we have started the process with $D_1$, hence $C'\cap S_1$ dominates $x_1$.

Therefore we can assume that all the sets $C'\cap S_i$ are set covers where $C'$ has size at most $|C|+2n+2n^2$. Since there are at most $|C|$ vertices of $C'$ which are in $S_1\cup \ldots \cup S_\ell$, it means that an identifying code with $|C|=SOL_{IC}$ vertices of $G_{IC}$ gives a solution of set cover with $SOL_{SC} \leq \frac{SOL_{IC}}{\ell}$ vertices.

\

\paragraph{Concluding on the size of the solutions}
 Assume now  that $SOL_{IC}\leq r \cdot OPT_{IC}$ for some value $r$, then:

$$SOL_{SC}\leq \frac{SOL_{IC}}{\ell} \leq \frac{r \cdot OPT_{IC}}{\ell} \leq \frac{r((2n^2-1)  OPT_{SC}+2n+2n^2)}{2n^2-1}\leq 2r \cdot OPT_{SC}$$

As before for discriminating codes, it achieves the proof of Theorem~\ref{thm:C4bip}. \vspace{10pt}
\end{proofof}

An \emph{edge colouring} of a graph with $k$ colours is a function $c: E\to \{1, \ldots, k\}$ such that no two edges sharing an endpoint are given the same colour, that is $c(uv)\neq c(uv')$ for every pair of edges $uv, uv'$.

\

\begin{proofof}{Claim \ref{C4-free}}
We first need to convince ourselves that $|S| \leq n^2$ in the instance $(X,S)$ of \scone. Indeed every pair of elements of $X$ appears in at most one $s\in S$, thus 
$|S|\leq \frac{n(n+1)}{2}\leq n^2$ (one for each pair plus $n$ additional singletons).

Now the idea of the proof is the following: we will represent our problem using a clique on $2n^2$ vertices. The vertices of the clique represent vertices of $Z$ and edges of the clique represent vertices of $S_1\cup \dots \cup S_\ell$.

 The edges of $K_{2n^2}$ can be partitioned into $2n^2-1$ perfect matchings, or equivalently there exists an edge colouring $c$ of $K_{2n^2}$ with $2n^2-1$ colours such that each colour class contain $n^2$ edges. Then label the vertices by $\{z_1, \ldots, z_{2n^2}\}$ and create a set $S'=S_1'\cup \dots \cup S'_{2n^2-1}$ of  elements $s_{j,k}$ with $j,k\in \mathbb{N}$ according to the following rule:
$$S'_i=\{s_{j,k} | c(z_j z_k)=i\} \textrm{ for every } i\in \{1, \ldots 2n^2-1\}$$
Observe that $|S'_i|=n^2$. Since a colour class is a matching, the indices of every pair of edges in a same colour class are pairwise distinct. In other words, there cannot be two vertices $s_{j,k}$ and $s_{j',k'}$ in the same set $S'_i$ if one of $j,k$ is equal to one of $j',k'$.
Now choose arbitrarily $|S|\leq n^2$ vertices in $S'_i$ to form $S_i$.
\end{proofof}

\section{Constant approximation algorithm for interval graphs}\label{sec:interval}

We now focus on the class of interval graphs and provide a constant factor approximation algorithm for \idcprob\ via a linear programming approach. More precisely we show that \idcprob\ has a $6$-approximation algorithm. The existence of a constant approximation algorithm was left open in~\cite{F13}.

Let us recall that an \emph{interval graph} is a graph which can be represented as an intersection of segments in the real line. We put an arbitrary order on the real line. The \emph{begin date of an interval $x$} is the first point $p$ of the real line (in the order) such that $p \in x$. The \emph{end date of $x$} is the last point which is in $x$. By abuse of notations, we will denote by $v$ both the vertex of the graph and the interval in the representation on the real line. Note that there exist many representations as intersections of segments for a same interval graphs, we choose arbitrarily one of them which can be found in linear time~\cite{BK}.

Let $G$ be an interval graph together with an interval representation. We denote its vertex set by $\{1, \ldots , n\}$ and $U \bigtriangledown U'$ stands for the symmetric difference of $U$ and $U'$ for $U,U'\subseteq V$.
Let us express \idcprob\ in terms of an integer program $P$, where $x_i$ is the decision variable corresponding to vertex $i$:

\renewcommand{\arraystretch}{2}
\begin{tabular}{lc}
\emph{Integer program} &$P$ \\
\emph{Objective function: } &$\min \displaystyle\sum_{i\in V} x_i$ \\
\emph{Separation constraint: } & $\displaystyle\sum_{i\in N[j] \bigtriangledown N[k]} x_i  \geq 1 \quad \forall j\neq k\in V $\\
\emph{Domination constraint: } & $\displaystyle\sum_{i\in N[j]} x_i \geq 1 \quad \forall j\in V$  \\ 
\emph{Integrality}: &$ x_i\in \{0,1\} \quad \forall i\in V $\\
\end{tabular}

Let us denote by $P^*$ the linear programming relaxation of $P$, where the integrality constraint is replaced by a non-negativity constraint $x_i\geq 0$, $\forall i\in V$. Recall that even if an integer linear program cannot be solved in polynomial time, its fractional relaxation can on the contrary be solved, using for instance the ellipsoid method. Our goal is to construct a feasible solution for $P$ of value at most $6\cdot OPT(P^*)$.
To achieve this goal, we decompose $P$ into two subproblems:

\noindent \begin{minipage}[t]{0.5\linewidth}
\centering
\noindent \begin{tabular}[t]{c}
$P_{inter}$ \\
$\min \displaystyle\sum_{i\in V} x_i$ \\
{\mathversion{bold} $\forall jk\in E$} $ \displaystyle\sum_{i\in N[j] \bigtriangledown N[k]} x_i  \geq 1$ \\ 
\parbox{0.8\linewidth}{\noindent (Separation constraints for intersecting pairs)}\\
\phantom{$\displaystyle\sum_{i\in N[j]} x_i \geq 1 \quad \forall j\in V$ 
(Domination constraints)}\\
$ x_i\in \{0,1\} \quad \forall i\in V $\\
\end{tabular}
\end{minipage}\vline
\begin{minipage}[t]{0.5\linewidth}
\centering
\begin{tabular}[t]{c}
$P_{disj}$ \\
$\min \displaystyle\sum_{i\in V} x_i$ \\
 {\mathversion{bold} $\forall jk\notin E$}  $ \displaystyle\sum_{i\in N[j] \bigtriangledown N[k]} x_i \geq 1$\\ 
\parbox{\linewidth}{(Separation constraints for non-intersecting pairs)}\\
$\displaystyle\sum_{i\in N[j]} x_i \geq 1 \quad  \forall j\in V $ 
(Domination constraints)\\
$ x_i\in \{0,1\} \quad \forall i\in V  $\\
\end{tabular}
\end{minipage}

The reason why intersecting intervals play a special role is that the symmetric difference of $N[j]$ and $N[k]$ can be expressed in some sense by an union of 2 intervals which greatly helps. We denote by $P_{inter}^*$ (resp. $P_{disj}^*$) the linear programming relaxation of $P_{inter}$ (resp. $P_{disj}$).

\

\begin{lemma}\label{P inter}
Given an optimal solution $(x_1^*, \ldots, x_n^*)$ of $P_{inter}^*$ of cost $OPT(P_{inter}^*)$, there exists a polynomial time algorithm that computes a solution of $P_{inter}$ of value at most $4\cdot OPT(P_{inter}^*)$. 
\end{lemma}

\

\begin{proof}
We follow the ideas of the proof of \cite{Gaur02} where the problem is translated in terms of Rectangle Stabbing Problem. Note that our problem can also be viewed as the transversal of 2-intervals (union of 2 intervals) and, in this respect, topological bounds can be found in \cite{K}, even if this does not provide an approximation algorithm. 

Le $x^*=(x_1^*, \ldots, x_n^*)$ be an optimal solution of $P_{inter}^*$. For every $jk\in E$, Figure \ref{fig: left right} shows how to partition $N[j] \bigtriangledown N[k]$ into two parts $L_{jk}$ (stands for Left) and $R_{jk}$ (stands for Right). The set $L_{jk}$ is composed of the intervals that end between the begin dates of $j$ and $k$, and $R_{jk}$ is composed of the intervals that begin between the end dates of $j$ and $k$. $L_{jk}$ and $R_{jk}$ are obviously disjoint subsets.

\begin{figure}
\begin{center}
\includegraphics[scale=1.2]{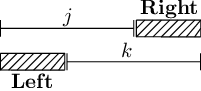} \hspace{2cm}
\includegraphics[scale=1.2]{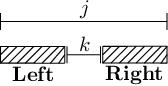}
\end{center}
\caption{Given two intersecting intervals $j$ and $k$, one can construct two areas \textbf{Left} and \textbf{Right} partitioning $N[j]\bigtriangledown N[k]$ between $L_{jk}$ the set of intervals that end in \textbf{Left}, and $R_{jk}$ the set of intervals that begin in \textbf{Right}. This figure shows how to find \textbf{Left} and \textbf{Right} depending on the configuration of $j$ and $k$: either one is included in the other, or not. }
\label{fig: left right}
\end{figure}

Let us now define two subsets of vertices $L$ and $R$ as follows:
\[L =\left\lbrace  jk\in E \ \vline \quad  \sum_{i\in L_{jk}} x_i^*\geq \frac{1}{2}\right\rbrace \qquad \textrm{and} \qquad R=\left\lbrace  jk\in E \ \vline \quad  \sum_{i\in R_{jk}} x_i^*\geq \frac{1}{2}\right\rbrace \]

Since all the constraints of $P_{inter}^*$ are satisfied by $(x_1^*, \ldots, x_n^*)$ and since for every edge $jk$, we have $L_{jk} \cup R_{jk}= N[j] \bigtriangledown N[k]$, all the edges are in $L$ or in $R$ (they can be in both of them).
Based on this, we define now the following two integer linear programs: 

\bigskip

\noindent \begin{minipage}{0.5\linewidth}
\centering
\noindent \begin{tabular}{c}
$P_{L}$ \\
$\min \displaystyle\sum_{i\in V} x_i$ \\
{\mathversion{bold} $\forall jk\in L$} $\displaystyle\sum_{i\in L_{jk}} x_i \geq 1 $\\ 
$ x_i\in \{0,1\} \quad \forall i\in V $\\
\end{tabular}
\end{minipage}\vline
\begin{minipage}{0.5\linewidth}
\centering
\begin{tabular}{c}
$P_{R}$ \\
$\min \displaystyle\sum_{i\in V} x_i$ \\
{\mathversion{bold} $\forall jk\in R$} $\displaystyle\sum_{i\in R_{jk}} x_i \geq 1$\\ 
$ x_i\in \{0,1\} \quad \forall i\in V $\\
\end{tabular}
\end{minipage}

\bigskip

According to the previous notations, we denote by $P_L^*$ (resp. $P_R^*$) the linear programming relaxation of $P_L$ (resp. $P_R$). Consider now the 0/1 matrix $M$ obtained from $P_L^*$, where each row represents an edge $jk\in L$, each column represents a vertex $i\in V$, and $M_{jk,i}=1$  if $i\in L_{jk}$, 0 otherwise. By sorting the vertices of $V$ (and thus, the columns of the matrix) by interval end date, the 1's on each row become consecutive. Indeed the 1's on the line of the constraint $jk$ correspond to intervals that end between the begin dates of $j$ and $k$ which are obviously consecutive if we sort the intervals by end date.
A matrix which has consecutive 1's on each row is said to have the \emph{interval property}. Such a matrix is totally unimodular (which means that all the squared determinants of the matrix have values $-1,0$ or $1$) and this implies that there is an optimal solution of $P_L^*$ where all the variables are integer. In particular,  $OPT(P_L^*)=OPT(P_L)$. (See~\cite{Martin99,Schrijver03} for more details about totally unimodular matrices and their use in linear programming.) 

The same holds for $P_R^*$ by sorting vertices by interval begin date. Solving $P_L^*$ and $P_R^*$ can be done in polynomial time, this gives us integer solutions $(x_1^L, \ldots, x_n^L)$ for $P_L^*$ and $(x_1^R, \ldots, x_n^R)$ for $P_R^*$ and setting $x_i=x_i^L+x_i^R$ builds a feasible solution for $P_{inter}$ of objective value $SOL(P_{inter})=OPT(P_L^*)+OPT(P_R^*)$.
Observe now that $(2x_1^*, \ldots, 2x_n^*)$ is a feasible solution for both $P_L^*$ and $P_R^*$ so $OPT(P_L^*)\leq 2\cdot OPT(P_{inter}^*)$ and $OPT(P_R^*)\leq 2 \cdot OPT(P_{inter}^*)$. This concludes the proof by $SOL(P_{inter})\leq 4 \cdot  OPT(P_{inter}^*)$.
\qquad \end{proof}

\

Let us now focus on the second subproblem:

\

\begin{lemma}\label{P disj}
Given the interval representation of $G$, one can compute in polynomial time a feasible solution for $P_{disj}$ of size at most $2 \cdot OPT(P_{disj}^*)$.
\end{lemma}

\

\begin{proof}
We construct a set $S$ of intervals in the following way. Initially, set $S=\emptyset$. While $V$ is not empty, do the following: select the interval $v$ that ends first ; put it in $S$ and remove $N[v]$ from $V$. Once $V$ is empty, output $S$.

Observe that this algorithm compute a maximal (with respect to inclusion) independent set $S$ with the property that for every vertex $v\in V$, there exists $s\in S$ such that the end date of $s$ is in the interval $v$ ($s$ is the vertex selected at the same step as $v$ was deleted). We now claim that on the one hand, $S$ is a feasible solution for $P_{disj}$ and on the other hand that $|S|\leq 2\cdot OPT(P_{disj}^*)$.

Let $j<k\in V, jk\notin E$. Up to symmetry, suppose that the interval $j$ starts before $k$ (and thus, ends before $k$ starts). Then there exists $s\in S$ such that the end date of $s$ is in $j$, implying $s\in N[j]\setminus N[k]$. As a maximal independent set, $S$ is also a dominating set so $S$ is a feasible solution for $P_{disj}$ with objective value $\alpha\in \mathbb{N}$.

Let us number $\{s_1, \ldots, s_\alpha\}$ the elements of $S$ by order of interval end date. Then observe that for every $i\in V$, there exist at most two distinct indices $j$ such that $i$ is in $N[s_j] \bigtriangledown N[s_{j+1}]$. Indeed, this happens if and only if $i$ begins between the end date of $s_j$ and the end date of $s_{j+1}$, or $i$ ends between the begin date of $s_j$ and the begin date of $s_{j+1}$. Then consider an optimal solution $(x_1^*, \ldots, x_n^*)$ of $P_{disj}^*$, we can derive:

\[ 2 \cdot OPT(P_{disj}^*)=\sum_{i\in V} 2\cdot x_i^* \geq \sum_{j=1}^{\alpha}\sum_{i\in N[s_j] \bigtriangledown N[s_{j+1}]} x_i^* \geq \alpha \ .\]
\qquad \end{proof}

\

\begin{theorem}\label{thm:interval}
There exists a polynomial time 6-approximation algorithm for \idcprob\ on interval graphs.
\end{theorem}

\

\begin{proof}
   \begin{algorithm}[ht!] \label{algo intervalle}
    \caption{6-Approximation algorithm for \idcprob\ in interval graphs}
  
    \KwIn{An interval graph $G=(V,E)$}
    \KwOut{An identifying code $C$ of size at most $6\cdot OPT$}
    \Begin{Compute an interval representation of $G$\;}
    \Begin(\texttt{ // Computation of the solution for} $P_{inter}$){
      Solve $P_{inter}^*$\;
      Compute $L$ and $R$\;
      Solve $P_L^*$ and $P_R^*$\;
      Set $S_{inter}$ to be the union of the solution to $P_L^*$ and $P_R^*$\;
      }
   
      \Begin(\texttt{ // Computation of the solution for} $P_{disj}$){        
        $S_{disj}=\emptyset$\;
        \While{ $V\neq \emptyset$}{
        Select the interval $v$ that ends first\;
        Add $v$ to $S_{disj}$\;
        Remove $N[v]$ from $V$\;
        }}
        \Return {$C=S_{inter}\cup S_{disj}$}
    
  \end{algorithm}
By Lemmata \ref{P inter} and \ref{P disj}, we can construct in polynomial time a solution $S_{inter}$ for $P_{inter}$ and a solution $S_{disj}$ for $P_{disj}$ of cost respectively at most $4 \cdot OPT( P_{inter}^*)$ and $2 \cdot OPT(P_{disj}^*)$. 
The set $S_{inter}\cup S_{disj}$ gives a feasible solution for $P$ of cost at most $6 \cdot OPT(P^*)$
 thus as most $6 \cdot OPT(P)$. Algorithm \ref{algo intervalle} sums up the different steps of the approximation algorithm. One can check that this algorithm runs in polynomial time since computing an interval representation of an interval graph can be done in linear time \cite{BK} and 
 solving  a linear programming relaxation can also be done in polynomial time using the ellipsoid method for instance.
\qquad \end{proof}

\

Observe that the bound on the cost of the solution is in fact $6 \cdot OPT(P^*)$ which is slightly  tighter than $6 \cdot OPT(P)$.

\bibliographystyle{plain}
\bibliography{biblio}

\end{document}